\documentclass[a4paper]{article}
\usepackage{epsfig}
\usepackage{amsmath}
\usepackage{amssymb}
\usepackage{color}
\usepackage{url}
\usepackage{wrapfig}
    
\usepackage[margin=1in]{geometry}
\usepackage{enumerate}
\usepackage{amssymb}
\usepackage{amsthm}
\usepackage{algorithmicx}
\usepackage[noend]{algpseudocode}
\usepackage{graphicx}
\usepackage{stmaryrd}
    
\usepackage[ruled,vlined]{algorithm2e}
    
\usepackage[final,mode=multiuser]{fixme}

\newcommand{\Occ}{\mathit{Occ}}

\newcommand{\LZ}{\mathsf{LZ}}

\newcommand{\LZRR}{\mathsf{LZRR}}

\newcommand{\LZr}{\mathsf{LZ'}}

\newcommand{\LPF}{\mathsf{LPF}}
\newcommand{\LNF}{\mathsf{LNF}}

\newcommand{\LZOR}{\mathsf{LZOR}}

\newcommand{\LCE}{\mathsf{lcp}}
\newcommand{\LCP}{\mathsf{LCP}}

\newcommand{\SA}{\mathsf{SA}}
\newcommand{\ISA}{\mathsf{ISA}}
\newcommand{\LP}{\mathsf{LP}}
\newcommand{\LF}{\mathsf{LF}}

\newcommand{\union}{\mathit{Union}}
\newcommand{\find}{\mathit{Find}}
\newcommand{\makeset}{\mathit{MakeSet}}

\newcommand{\nil}{\mathsf{NIL}}

\newcommand{\jmax}{j_{\mathit{max}}}
\newcommand{\kmax}{k_{\mathit{max}}}
\newcommand{\ellmax}{\ell_{\mathit{max}}}

\newcommand{\vroot}{\mathit{source}}

\newtheorem{theorem}{Theorem}

\newtheorem{lemma}[theorem]{Lemma}
\newtheorem{corollary}[theorem]{Corollary}

\newlength{\figurewidth}
\newlength{\smallfigurewidth}

\begin{document}
    
\title
{\large
\textbf{LZRR: LZ77 Parsing with Right Reference}
}    
\author{%
Takaaki Nishimoto$^{\ast}$ and Yasuo Tabei$^{\ast}$\\[0.5em]
{\small\begin{minipage}{\linewidth}\begin{center}
\begin{tabular}{ccc}
$^{\ast}$RIKEN Center for Advanced Intelligence Project & \hspace*{0.5in} \\
\url{{takaaki.nishimoto, yasuo.tabei}@riken.jp}
\end{tabular}
\end{center}\end{minipage}}
}

\maketitle
\thispagestyle{empty}

\begin{abstract}
Lossless data compression has been widely studied in computer science. 
One of the most widely used lossless data compressions is \emph{Lempel-Zip}~(LZ) 77 parsing, which achieves a high compression ratio. 
\emph{Bidirectional} (a.k.a. \emph{macro}) 
parsing is a lossless data compression and computes a sequence of phrases copied from another substring (\emph{target phrase}) on either the left or the right position in an input string.  
Gagie et al.~(LATIN 2018) recently showed that 
a large gap exists between the number of smallest bidirectional phrases of a given string and that of LZ77 phrases.
In addition, finding the smallest bidirectional parse of a given text is NP-complete. 
Several variants of bidirectional parsing have been proposed thus far, but 
no prior work for bidirectional parsing has achieved high compression that is smaller than that of LZ77 phrasing for any string. 
In this paper, we present the first practical bidirectional parsing named \emph{LZ77 parsing with right reference (LZRR)}, in which the number of LZRR phrases is theoretically guaranteed to be smaller than the number of LZ77 phrases.
Experimental results using benchmark strings show the number of LZRR phrases is approximately five percent smaller than that of LZ77 phrases. 
\end{abstract}
\section{Introduction}
Lossless data compression has been widely studied in computer science. 
One of the most widely used lossless data compressions is \emph{Lempel-Zip}~(LZ) 77 parsing~\cite{LZ76}, which compresses a given string by computing a sequence of phrases copied from the longest substring on the left position in an input string.
LZ77 parsing has a long research history, with the first paper on it published in 1976~\cite{LZ76}. 
Many LZ77's extensions have since been proposed ~(e.g.,~\cite{DBLP:conf/dcc/KreftN10,DBLP:journals/tcs/Jez16,DBLP:journals/tcs/Rytter03}), 
and LZ77 parsing achieves the smallest compression ratio among them.

\emph{Bidirectional} (a.k.a. \emph{macro}) parsing~\cite{DBLP:journals/jacm/StorerS82} is a lossless data compression and computes a sequence of phrases copied from another substring (\emph{target phrase}) on either the left or right position in an input string.
Each set of LZ77 phrases is convertible into a set of bidirectional phrases, and the number of phrases in the smallest bidirectional parsing is less than that of LZ77 phrases. 
Gagie et al.~\cite{DBLP:conf/latin/GagieNP18} recently showed the number of LZ77 phrases $z$ representing an input string of length $n$ can be tightly bounded by the smallest number of bidirectional phrases $b^*$ representing the same string as $z = O(b^{*} \log(n/b^{*}))$, 
which suggests that a large gap exists between $b^{*}$ and $z$.
In addition, finding the smallest bidirectional parse of a given text is NP-complete~\cite{DBLP:journals/jacm/StorerS82}. 
Thus, an important open challenge is to 
develop a polynomial time bidirectional parsing such that the number of bidirectional phrases is smaller than that of LZ77 phrases.  

Several variants of bidirectional parsing have been proposed thus far. 
\emph{Lex-parsing}~\cite{DBLP:journals/corr/abs-1803-09517} is a bidirectional parsing that computes a sequence of bidirectional phrases that each occurred previously on a suffix array of a string.
The number of phrases $v$ in the lex-parsing is bounded by $v = O(b^{*} \log (n/b^{*}))$~\cite{DBLP:conf/latin/GagieNP18}.
Although the lex-parsing is effective for most benchmark strings (i.e., phrases $v$ is very close to $z$) in practice, it can fail to compress some strings (i.e., $v$ is much larger than $z$)~\cite{DBLP:journals/corr/abs-1803-09517}. 
\emph{Lcpcomp}~\cite{DBLP:conf/wea/DinklageFKLS17} and 
a bidirectional parsing using Burrows-Wheeler transform (BWT)~\cite{DBLP:conf/latin/GagieNP18} have also been proposed, and 
they never have fewer phrases than lex-parse~\cite{DBLP:journals/corr/abs-1803-09517}.
Kempa and Prezza proposed a parsing algorithm for computing 
the bidirectional parse of an input string for a given \emph{string attractor} of the string~\cite{DBLP:conf/stoc/KempaP18}. 
The number of the bidirectional phrases is bounded by $O(\gamma \log (n/\gamma))$, 
where $\gamma$ is the size of the string attractor. 
Let $\gamma^{*}$ be the size of the smallest string attractor for a given string.
Then $b^{*} = O(\gamma^{*} \log (n/\gamma^{*}))$ holds~\cite{DBLP:conf/stoc/KempaP18}. 
In addition, finding the smallest string attractor of a given string is also NP-complete~\cite{DBLP:conf/stoc/KempaP18}.
In summary, no prior bidirectional parsing achieves high compression that is smaller than that of LZ77 phrasing for any string. 

In this paper, we present the first practical bidirectional parsing named \emph{LZ77 parsing with right reference (LZRR)} 
in which the number of LZRR phrases is always smaller than the number of LZ77 phrases by a large margin. 
LZRR is a polynomial time algorithm that greedily computes phrases from a string in the left-to-right order the same as LZ77. 
The main difference between LZRR and LZ77 is the way to compute their phrases. 
Whereas LZ77 parsing chooses the longest substring occurring previously as a phrase, 
LZRR parsing uses not only previous occurrences of each phrase but also subsequent occurrences~(i.e., it chooses the longest substring occurring previously or subsequently as a phrase). 
For this reason, the number of LZRR phrases is theoretically guaranteed to be no more than that of LZ77 phrases. 
Experimental results using benchmark datasets show the number of LZRR phrases is approximately five percent smaller than that of LZ77 phrases. 
\section{Preliminaries} \label{sec:preliminary}
Let $\Sigma$ be an ordered alphabet of size $\sigma$, 
$T$ be a string of length $n$ over $\Sigma$ and 
$|T|$ be the length of $T$.
Let $T[i]$ be the $i$-th character of $T$ and 
$T[i..j]$ be the substring of $T$ that begins at position $i$ and ends at position $j$. 
$T[i..]$ denotes the suffix of $T$ beginning at position $i$, i.e., $T[i..n]$. 
Let $T^R$ be the reversed string of $T$, i.e., $T^R = T[n] T[n-1] \cdots T[1]$. 

$\Occ(T,s)$ denotes all the occurrence positions of string $s$ in string $T$, i.e., 
$\Occ(T, s) = \{i \mid s = T[i,i+|s|-1], 1 \leq i \leq n-|s|+1\}$. 
Let $\LCE(i, j)$ be the length of the longest common prefix~(LCP) of $T[i..]$ and $T[j..]$. 
For two strings $x$ and $y$, $x \prec y$ represents that 
$x$ is lexicographically smaller than $y$, we write $x \prec y$. 
Similarly, for a string $z$,  $x \preceq_{z} y$ represents that  the LCP of $x$ and $z$ is equal to or longer than that of $y$ and $z$. 
For example, $aab \preceq_{aac} ab$. 
 
Our model of computation is a unit-cost word RAM with a machine word size of $\Omega(\log_2 n)$ bits. 
We evaluate the space complexity in terms of the number of machine words. A bitwise evaluation of 
space complexity can be obtained with a $\log_2 n$ multiplicative factor. 
\subsection{Arrays}

\emph{Suffix array} $\SA$, \emph{inverse suffix array} $\ISA$, \emph{LCP array} $\LCP$, \emph{longest previous factor array} $\LPF$, and \emph{sorted suffix array} $\SA_{i}$ are integer arrays of length $n$ for a string $T$, respectively.
$\SA$ is the permutation of $[1..n]$ such that $T[SA[1]..] \prec \cdots \prec T[SA[n]..]$ holds. 
$\ISA$ is the permutation of $[1..n]$ such that $\SA[\ISA[i]] = i$ holds for any $i \in \{1,2,...,n\}$. 
$\LCP[1] = 0$ and $\LCP[i] = \LCE(\SA[i], \SA[i-1])$ for $i \in \{ 2, 3, \ldots, n \}$. 
$\LPF[i]$ stores the length of the longest prefix of $T[i..]$ occurring previously; that is 
$\LPF[1] = 0$ and $\LPF[i] = \max \{ \LCE(i, p) \mid p \in \{ 1..i-1 \} \}$, where 
$\max$ returns the maximal element of a given set. 
$\SA_{k}$ is the sorted starting 
positions of suffixes in decreasing order for the length of the LCP with $T[k..]$. 
Formally, for an integer $k \in \{ 1, 2, \ldots, n \}$, 
$\SA_{k}$ is a permutation of $[1..n]$ such that $T[SA_{k}[1]..] \preceq_{T[k..]} \cdots \preceq_{T[k..]} T[SA_{k}[n]..]$. 
$\SA_{k}$ is not unique when there exist two positions $i$ and $j$ such that $\LCE(k, i) = \LCE(k, j)$. 

For $T=abababaabb$, 
$\SA = 7, 5, 3, 1, 8, 10 , 6 , 4, 2, 9$, $\ISA = 4, 9, 3, 8, 2, 7,  1, 5, 10, 6$, 
$\LCP = 0, 1, 3, 5, 2, 0, 1, 2, 4, 1$, $\LPF = 0, 0, 5, 4, 3, 2, 1, 2, 1, 1$, and 
$\SA_{1} = 1, 3, 5, 8, 7, 10, 6, 4, 2, 9$. 
\subsection{Union-find data structure}
\emph{Union-find} is a data structure for disjoint sets and supports the following operations for disjoint set $\mathcal{D}$: $\makeset$, $\union$, $\find$.
$\makeset$ adds element $\{ m + 1 \}$ into $\mathcal{D}$ and returns the integer where $m$ is the cardinality of  $\mathcal{D}$. 
$\union(x, y)$ merges two sets $X, Y \in \mathcal{D}$ containing $x$ and $y$, respectively; it adds a new set $X \cup Y$ into $D$; it removes $X$ and $Y$ from $\mathcal{D}$. 
The $\find(x)$ returns the id of the set containing $x$ in $\mathcal{D}$. 
The union-find data structure performs $\makeset$, $\union$, $\find$ operations in $O(m + p + q \alpha_{p+q}(p))$ time, 
while using $O(m)$ space~\cite{DBLP:journals/jacm/Tarjan75}, where $p$ and $q$ are the numbers of $\union$ and $\find$ operations, respectively,  
and $\alpha_{k}$ is the inverse of the $k$-th row of Ackermann function.

\subsection{Bidirectional phrases and partial bidirectional phrases}
\emph{Bidirectional phrases}~(BP)~\cite{DBLP:journals/jacm/StorerS82} of string $T$ is a partition of $T$ as 
substrings (phrases) $B=f_{1}, f_{2}, \ldots, f_{b}$ such that 
each $f_{i}=T[s_i..s_i+\ell-1]$ is (i) either copied from another substring $T[t_i..t_i+\ell-1]$ (\emph{target phrase}) with $s_i \neq t_i$, which can overlap $T[s_i..s_i+\ell-1]$, or 
(ii) an explicit character (\emph{character phrase}), i.e., $f_{i}=T[s_i]$. 
Target phrase $f_{i}$ is denoted as a pair $\langle t_{i}, |f_{i}| \rangle$  of 
the \emph{reference position} $t_{i}$ and the length $|f_{i}|$ of $f_{i}$. 
The substring $T[t_{i}..t_{i} + |f_{i}|-1]$ is called the \emph{reference string} of $f_{i}$. 

The original string $T$ can be recovered from BP $B$ by 
referring to a finite number of phrases from each $f_i$ in $B$. 
If an infinite loop of phrases referred from any $f_i$ exists, 
the original string $T$ cannot be recovered from $B$.
If $T$ can be recovered from $B$, $B$ is said to be a \emph{valid} BP of $T$; 
otherwise, $B$ is said to be \emph{invalid} BP of $T$.

The value of the phrase reached from position $x$ in $k$ iterations of references is formally defined as $g^{k} : \{1, \ldots, n\} \rightarrow \{1, \ldots, n \} \cup \Sigma$. 
For $k=0$, if $T[x]$ is a character phrase, $g^{0}(x) = T[x]$; otherwise $g^{0}(x) = t_{p} + (x - s_{p})$, 
where $p$ is the integer such that $s_{p} \leq x < s_{p+1}$ holds for $s_{b+1} = n+1$. 
For $k \geq 1$, we define $g^{k}(x)$ as follows:
\begin{eqnarray}
	g^{k}(x) = \left\{
	\begin{array}{ll}
		g^{k-1}(x) & \mbox{if } g^{k-1}(x)\in \Sigma,  \\
	    g^{0}(g^{k-1}(x)) & \mbox{otherwise}.
	\end{array}
	\right.
\end{eqnarray}
If $g^{n}[x] \in \Sigma$ holds, then there are no infinite loops of references containing $x$. 
Therefore, $B$ is valid BP of $T$ if $B$ has no infinite loops of references, i.e., 
$g^{n}[x] \in \Sigma$ holds for all $x$. 

For example, let $B = \langle 3, 2 \rangle, a, b, \langle 2, 3 \rangle$ and $B' = \langle 3, 2 \rangle, \langle 1, 2 \rangle, b, a, b$ be BPs of $T=ababbab$. 
Then $B$ is valid since $g^{1}(1), \ldots, g^{1}(7) = a, b , a, b, 4, a, b$ and $g^{2}(5)= b$. 
On the other hand, $B'$ is invalid since $g^{7}(1), \ldots, g^{7}(7) = 1, 2, 3, 4, b, a, b$.  

\emph{LZ77 phrases}~\cite{LZ76} of string $T$ are a specialization of BP and defined as 
the bidirectional phrases that are all selected from previously seen substrings.
Since there is no infinite loops of references on phrases, LZ77 phrases of $T$ are always valid BP of $T$. 
Formally, let $\LZ(T) = f_1, f_2, \ldots, f_z$ of $T$ be valid BP of $T$ such that 
$|f_{i}| = \max \{ 1, \LPF[s_{i}] \}$ for each $i \in \{1,...,z\}$. 

LZRR parsing gradually builds the valid BP from the start position of $T$ in the left-to-right order. 
A subsequence of the valid BP is called \emph{partial bidirectional phrases~(PBP)} 
and is defined as a BP $P = f_{1}, f_{2}, \ldots, f_{k}$ for a prefix of $T$ 
that can be copied from 
any substring of $T$, i.e., $t_{i} \in \{ 1, \ldots, n \} \setminus \{ s_{i} \}$ 
for all $i \in \{ 1, \ldots, k \}$ for a target phrase $f_{i}$, which avoids a self copy.

The concatenation of such PBP $P$ and every character phrase referred from $P$ can recover the prefix of $T$ with a finite number of references. 
Such PBP are called valid PBP, 
and other PBP are called invalid PBP. 
Formally, let $B_{P} = P \cdot f'_{1}, f'_{2}, \ldots, f'_{k'}$ be the concatenation of PBP $P$ and 
the remaining character phrases $f^\prime_{1}, f^\prime_{2}, \ldots, f^\prime_{k}$ equivalent to suffix $T[(n - k' + 1)..]$. 
$P$ is valid if $B_{P}$ is valid; otherwise $P$ is invalid.
For example, let $P = \langle 3, 2 \rangle, \langle 6, 2 \rangle$ be a PBP of $T=ababbab$. 
Then $B_{P} = \langle 3, 2 \rangle, \langle 6, 2 \rangle, b, a, b$.

The original string of a PBP can be recovered by iteratively referring to phrases starting from each target phrase in a finite number of times until the character phrase is found. 
Thus, the position of each character phrase can be seen as the \emph{source} for positions of target/character phrases.
Formally, for a PBP $P$ and position $x \in \{ 1, \ldots, n\}$ on $T$, 
$\vroot(P , x)$ returns source $y \in \mathcal{N}$ of $x$ in $B_{P}$, 
i.e., position $y$ satisfying either 
(i) $g^{k}(x) = y$ and $g^{k+1}(x) \in \Sigma$ for an integer $k$ or 
(ii) $x = y$ and $g^{0}(x) \in \Sigma$. 
For the above example, the source of the position $1$ is the position $6$ in $P$ since $g^{0}(1) = 3$, $g^{1}(1) = 6$ and $g^{2}(1) = a$.

\section{LZRR}\label{sec:algo}
A key idea of LZRR parsing is to compute the valid BP from an input text $T$ 
by gradually computing the valid PBP from the head of $T$ in the left-to-right order. 
LZRR parsing computes whole LZRR phrases initialized as zero phrase for an input string 
in two steps: 
(i) it computes candidates of the reference positions of 
the longest valid phrase following the current LZRR phrase; and 
(ii) it computes the valid (possibly character) phrase with 
the maximum length among extensions starting from those candidates.
Steps (i) and (ii) are iterated until whole LZRR phrases are computed.

LZRR parsing uses two major functions of LP and LF for steps (i) and (ii), respectively.  
Given a valid PBP $P$ of $T$, LP function $\LP(P)$ returns the longest valid phrase following $P$, i.e., the longest phrase $f$ such that $P \cdot f$ is a valid PBP of $T$. 
Given a valid PBP $P$ of $T$ and reference position 
$j \in \{ 1, \ldots, n \}$, LF function $\LF(P,j)$ returns 
the length of the longest valid phrase having reference position $j$ and following $P$, 
i.e., $\LF(P, j) = \max (\{ 0 \} \cup \{ \ell \mid \ell \in \{ 1, 2, \ldots, \LCE(i, j) \}, P \cdot \langle j , \ell \rangle \mbox{ is valid} \}$) 
where $i$ is the starting position of the phrase following $P$.
LZRR parsing computes LZRR phrases as the valid BP $\LZRR(T) = \LP(P_{0}), \ldots , \LP(P_{b-1})$ of $T$　
where $P_{p}$ is the first $p$ LZRR phrases for each $p \in \{ 0, 1, \ldots , b\}$ and 
$b$ is the number of LZRR phrases of $T$. 
The LZRR phrases of $T$ are not unique. 

For example, let $P_{1} =  \langle 3, 5 \rangle$ be the first LZRR phrase of $T=abababaababa$. 
$\LF(P_{1}, 1), $ $\ldots, \LF(P_{1}, 12) = 0, 0, 0, 0, 0, 0, 0, 0, 2, 0, 2, 0$. 
LZRR parsing chooses phrase $\langle 9, 2 \rangle$ or $\langle 11, 2 \rangle$ as the next one.

This paper shows the following two theorems. 
\begin{theorem}\label{theo:algo}
	For a given string $T$, LZRR parsing computes $\LZRR(T)$ in $O(n^2 \alpha_{n^2}(n^2))$ time using $O(n)$ working space. 
\end{theorem}
\begin{theorem}\label{theo:bound}
    $|\LZRR(T)| \leq  |\LZ(T^{R})|$ holds.
\end{theorem}

The LZRR parsing algorithm is presented in Section~\ref{sec:algo}.
Theorems~\ref{theo:algo} and~\ref{theo:bound} are shown in Section~\ref{sec:bound}.

\subsection{$\LP$ algorithm}\label{sec:lp}
\begin{algorithm}[t]
  \SetKwProg{Fn}{Function}{}{}
	Input:$P = f_{1}, \ldots , f_{p}$, Output: $\LP(P)$\;
  $(k, j) \leftarrow (1, \SA_{i}[1])$\;
  $(k, \jmax, \ellmax) \leftarrow (1, \nil, 0)$\;
  \While{true}{
      $j \leftarrow \SA_{i}[k++] $ \tcp{$i = |f_{1} \cdots f_{p}|+1$. }
      if $\LCE(i, j) \leq \ellmax$ then break\;
      $\ell \leftarrow \LF(P, j) $\;
      if $\ell > \ellmax$ then $(\jmax, \ellmax) \leftarrow (j, \ell)$\;
  }
  if $\ellmax > 0$ then \Return $\langle \jmax, \ellmax \rangle$ else \Return $T[i]$\;
  \caption{The $\LP$ algorithm.\label{algo:lp}}
\end{algorithm}
A straight forward computation of $\LP(P)$ is to compute reference position $\jmax$ such that 
$\LF(P, \jmax) = \max \{ \LF(P, 1), \ldots, \LF(P, n) \}$ and then compute 
$\ellmax = \LF(P, \jmax)$, which results in LZRR phrase $\langle \jmax, \ellmax \rangle$. 
This method takes $\Omega(n)$ time even if $\LF(P, j)$ can be computed in constant time for each position $j$.
Instead, we reduce the computation time of LF functions 
by leveraging the following fact: 
the length of the longest valid phrase of starting position $i$ and reference position $j$ is not larger than that of the LCP of $T[i..]$ and $T[j..]$. 
This fact suggests that after we find a phrase of length $\ell'$, 
we do not need to compute LF functions for any reference position $j$ such that 
the LCP of $T[i..]$ and $T[j..]$ is not longer than $\ell'$.
For an efficient computation, we sort reference positions in descending order 
with respect to the length of the LCP for $T[i..]$ and maintain those positions 
in the sorted suffix array $\SA_{i}$ of $i$. 
Then, we omit computing LF functions of reference positions on $\SA_{i}[\kmax+1..]$ for 
the left-most position $\kmax$ on $\SA_{i}$ such that 
the longest valid phrase starting at a reference position in $\SA_{i}[1..\kmax]$ is at least as long as the LCP of $T[i..]$ and $T[\SA_{i}[\kmax]..]$.
This is because $\jmax$ exists on $\SA_{i}[1..\kmax]$.
Thus, the following lemma holds. 
\begin{lemma}\label{lem:maxlf}
    Let $\ell_{k} = \max \{ \LF(P, \SA_{i}[1]), \ldots, \LF(P, \SA_{i}[k]) \}$ and 
    $\kmax$ be the left-most position on the $\SA_{i}$ such that $\ell_{\kmax} \geq \LCE(i, \SA_{i}[\kmax])$ holds. 
    Then $\ellmax = \ell_{\kmax}$ holds and $\SA_{i}[..\kmax]$ contains $\jmax$. 
\end{lemma}
\begin{proof}
See Appendix.
\end{proof}

Algorithm~\ref{algo:lp} shows the algorithm for computing $\LP(P)$ function and computes each LF function from the head of $\SA_{i}$.
When Algorithm~\ref{algo:lp} finds $\kmax$, 
it returns the current longest valid phrase. 

\subsection{$\LF$ algorithm}\label{sec:lf}
$\LF$ algorithm $\LF(P, j)$ finds the longest valid target phrase with reference position $j$ and following the PBP $P$ of $T$
by gradually extending the target phrase of length $1$ until it cannot find any reference string copying the target phrase. 
When PBP $P \cdot \langle j, \ell \rangle$  for $P$ and the target phrase $\langle j, \ell \rangle$ is computed one-by-one, it can include an infinite loop of references by a mutual reference of phrases. 
This is because PBP as a target phrase can be copied from the left and right reference strings. 
This can happen when for computing the extension $P \cdot \langle j, \ell \rangle$ the position of a target phrase in $P$ and $\langle j, \ell \rangle$ can be mutually reached with a finite number of references. 
The $\LP$ algorithm avoids such cases by using the union-find data structure built from PBP $P$. 

Each disjoint set in the union-find data structure includes string positions with the same source (character phrase) for PBP $P$. 
The union-find data structure is initialized as $n$ disjoint sets that all contain the unique position of the input string of length $n$.
If the union-find data structure for PBP $P \cdot \langle j, \ell-1 \rangle$ for $P$ and the target phrase $\langle j, \ell \rangle$ exists, the data structure for $P \cdot \langle j, \ell-1 \rangle$ can be updated by $\union(i+\ell-1, j+\ell-1)$ operation. 

The infinite loops of references can be detected using the find operation in the union-find data structure. 
When $P \cdot \langle j, \ell - 1 \rangle$ is a valid and the extension of 
starting position $i+\ell-1$ next to the PBP and reference position $j+\ell-1$ is computed, 
if $\find(i+\ell-1)$ is equal to $\find(j+\ell-1)$ if and only if infinite loops of references exist. 
$\LF$ algorithm checks this condition each time. 
Formally, the following corollary holds. 

\begin{corollary}\label{cor:da2} 
Let $Q = P \cdot \langle j, \ell \rangle$ be a valid PBP and 
$Q' = P \cdot \langle j, \ell + 1 \rangle$ be a PBP 
for an integer $\ell \in \{ 0, \ldots, \LCE(i,j)-1 \}$, 
and $\mathcal{D}_{P}$ be disjoint sets on $\{ 1, \ldots, n \}$ 
such that each set consists of all positions of the same source for a PBP $P$, 
where $P \cdot \langle j, 0 \rangle$ is $P$ and 
$i$ is the starting position of the last target phrase~(i.e., $\langle j, \ell \rangle$) in $Q$. 
(1) If $\find(i+\ell) \not = \find(j+\ell)$ holds on $\mathcal{D}_{Q}$, 
then $Q'$ is valid. Otherwise $Q'$ is invalid. 
(2) $\mathcal{D}_{Q'}$ is equal to the set created by $\union(i+\ell, j+\ell)$ on $\mathcal{D}_{Q}$.
\end{corollary}

\begin{algorithm}[t]
  \SetKwProg{Fn}{Function}{}{}
  Input: PBP $P$, union-find data structure for disjoint sets of $P$, reference position $j$\; 
  Output: $\LF(P, j)$\;
  $\ell \leftarrow 1$\;
  \While{$T[i+\ell-1] = T[j+\ell-1]$ and $\find(i + \ell -1) = \find(j + \ell -1) $}{
      $\union(i + \ell -1, j + \ell -1)$\;
      $\ell \leftarrow \ell + 1$ \;
  }
 \Return $\ell - 1$\;
  \caption{The $\LF$ algorithm.\label{algo:lf}}
\end{algorithm}
Algorithm~\ref{algo:lf} shows the algorithm for computing $\LF(P, j)$ function 
using Corollary~\ref{cor:da2} and the algorithm stated previously.
Thus, we can compute the length $\ellmax$ of 
the longest valid target phrase with reference position $j$ and following the PBP $P$
by $O(\ellmax)$ union and find operations on the given union-find data structure for $\mathcal{D}_{P}$. 

Note that we need to modify Algorithm~\ref{algo:lf} for $\LP(P)$ algorithm. 
This is because $\LF$ algorithms in our $\LP(P)$ algorithm need 
the same union-find data structure determined by the PBP $P$. 
On the other hand, 
the given union-find data structure is changed by union operations in Algorithm~\ref{algo:lf}.
By modifying Algorithm~\ref{algo:lf} using an additional union-find data structure, we can compute $\LF(P, j)$ without updating the given union-find data structure. Formally, the following lemma holds.

\begin{lemma}\label{lem:lf}
    Given the union-found data structure $L$ for $\mathcal{D}_{P}$, 
    we can compute $\LF(P, j)$ in $O(n)$ 
    working space by $O(\ellmax)$ $\find$ operations on $L$ and
    $O(\ellmax)$ union and find operations on 
    an additional union-find data structure $L'$ for $O(\ellmax)$ disjoint sets. 
    $L'$ is disposed after $\LF(P, j)$ is computed.
\end{lemma}
\begin{proof}
 See Appendix.
\end{proof}

\subsection{Computation of $\LZRR(T)$}\label{sec:complzrr} 
Since $\LF(P_{p}, j)$ algorithm for each $p \in \{ 0, \ldots, b-1 \}$ uses the union-find data structure for disjoint sets of the current LZRR phrases~(i.e., $P_{p}$), 
we update the union-find data structure when the $(p+1)$-th LZRR phrase is selected. 
This needs at most $|\LP(P_{p+1})|$ $\union$ operations by Corollary~\ref{cor:da2}.

\section{Theoretical analysis}\label{sec:bound}

\subsection{The proof of Theorem~\ref{theo:algo}}\label{sec:analysis}
We show that the working space of LZRR parsing is $O(n)$ space.
LZRR parsing needs two data structures: 
(1) the union-find data structures for $\LF$ algorithm and 
(2) the data structure to compute the sequence $W_{p}= \SA_{s_{p}}[1], \LCE(s_{p}, \SA_{s_{p}}[1]), \ldots, \SA_{s_{p}}[k_{p}], \LCE(i, \SA_{i}[k_{p}])$ for $\LP(P_{p-1})$ algorithm, where 
$k_{p}$ is $\kmax$ in $\LP(P_{p-1})$ algorithm and 
$s_{p}$ is the starting position of $p$-th LZRR phrase.

We can compute $W_{p}[1..k]$ in $O(k)$ time in an online manner 
using arrays of $\SA, \ISA$, and $\LCP$ for 
two integers $p$ and $k$~(See Appendix). 

$\SA$, $\ISA$, and $\LCP$ of a given a string $T$ can be constructed in $O(n)$ time and working space~\cite{DBLP:conf/icalp/KarkkainenS03,DBLP:conf/cpm/KasaiLAAP01}. 
Therefore, the second data structure can be constructed in $O(n)$ time and space, 
and the LZRR parsing algorithm runs in $O(n)$ working space.

Next, we show that the running time of LZRR parsing is $O(n^2 \alpha_{n^2}(n^2))$.
Let $G$ be the sequence of operations on disjoint-sets executed by LZRR parsing and 
$W$ be the sequence of $W_{1} \cdots W_{b}$, 
where $b$ is the number of phrases in $\LZRR(T)$.
Then the running time is the sum of the computation time for executing $G$ and computing $W$, and the prepossessing time of $\SA$, $\ISA$ and $\LCP$, which is $O(n)$. 

We show that $W$ can be computed in $O(n^2)$ time. 
For an integer $p \in \{ 1, \ldots, b \}$, 
$W_{p}$ can be computed in $O(k_{p}) = O(n)$ time.
This is because $k_{p} = |\Occ(T, f_{p})| \leq n$ holds since
$T[\SA_{i}[y]..]$ has $f_{p}$ as a prefix for all $y \in \{ 1 , \ldots , k_{p} \}$, 
where $f_{p}$ is the string represented by the $p$-th LZRR phrase. 
Thus, $|W| = O(n^2)$ since $b \leq n$. 
Hence $W$ can be computed in $O(n^2)$ time using the above online algorithm. 

We show that $G$ is performed in $O(n^2 \alpha_{n^2}(n^2))$ time.
$|G| = O(\sum_{p=1}^{b} (|f_{p}| \times k_{p}))$ holds because 
$\LP(P_{p})$ performs $O(k_{p+1} \times |f_{p+1}|)$ union and find operations for 
$p \in \{ 0, \ldots, b-1 \}$.
Since $|f_{1}| + \cdots + |f_{n}| = n$ and $|\Occ(T, f_{p})| \leq n$ for all $p$, 
$|G| = O(n^2)$ holds. 
Therefore, $G$ is performed in $O(n^2 \alpha_{n^2}(n^2))$ time by union-find data structures. 

As a result, we can compute $\LZRR(T)$ in $O(n^2 \alpha_{n^2}(n^2))$  time and $O(n)$ working space.  

\subsection{The proof of Theorem~\ref{theo:bound}}
We define two BPs $\LZr(T)$ and $\LZOR(T)$ for Theorem~\ref{theo:bound} and 
show three formulas: (1)$|\LZr(T)| = |\LZ(T)|$, (2)$|\LZRR(T)| \leq |\LZOR(T)|$, and (3) $|\LZOR(T)| = |\LZ'(T^{R})|$. 
Theorem~\ref{theo:bound} clearly holds in (1), (2), and (3), 
i.e., $\LZRR(T) \leq \LZ(T^{R})$. 
The detailed proofs are in Appendix. 

\smallskip
\noindent \textbf{The proof of $|\LZr(T)| = |\LZ(T)|$.}
$\LZr(T) = f_{1}, \ldots, f_{k}$ parses greedily $T$ in the right-to-left order such that each phrase is the longest substring occurring previously~(left) in $T$.

A key idea of this proof is that if $\LZr(T)$ chooses a substring as an $\LZr$ phrase, 
then there exists an LZ phrase starting at a position on the $\LZr$ phrase and including the ending position of the $\LZr$ phrase. 
This is because the $\LZr$ phrase occurs previously in $T$ and the LZ phrase is the longest substring occurring previously in $T$.  
Since the fact holds for every $\LZr$ phrase, $|\LZ(T)| \leq |\LZr(T)|$ holds. 
Conversely, 
if $\LZ(T)$ chooses a substring as an LZ phrase,  
then there exists an $\LZr$ phrase starting at a position on the LZ phrase and including the starting position of the LZ phrase. 
This is because the LZ phrase occurs previously in $T$ and 
the $\LZr$ phrase is the longest substring occurring previously in $T$. 
Since this fact holds for every LZ phrase, $|\LZr(T)| \leq |\LZ(T)|$ holds. 
Therefore, $|\LZr(T)| = |\LZ(T)|$ holds.

\smallskip
\noindent \textbf{The proof of $|\LZRR(T)| \leq |\LZOR(T)|$.}
$\LZOR(T) = f_{1}, \ldots, f_{k}$ parses $T$ in the left-to-right order such that each phrase is the longest substring occurring subsequently in $T$. 

A key idea of this proof is that if $\LZOR(T)$ can choose a substring at a position as an LZOR phrase 
then $\LZRR(T)$ also can choose the substring as an LZRR phrase. 
This is because candidate phrases with right reference positions are always valid phrases in LZRR parsing. 
Since the fact holds for every position on $T$, $|\LZRR(T)| \leq |\LZOR(T)|$ holds. 

\smallskip
\noindent \textbf{The proof of $|\LZOR(T)| = |\LZ'(T^{R})|$.}
Parsing a string in the left-to-right order using 
the longest substring occurring subsequently in the string is equal to 
parsing the reversed string in the right-to-left order using 
the longest substring occurring previously in the reversed string. 
Thus, $|\LZOR(T)| = |\LZr(T^R)|$ holds.

\section{Experiments}\label{sec:exp}
In this section, 
we demonstrate the effectiveness of LZRR parsing with benchmark strings. 
We used two types of strings of pseudo-real and real repetitive collections 
in the Pizza \& Chili corpus downloadable from \url{http://pizzachili.dcc.uchile.cl}.
We compared our LZRR parsing with LZ77 parsing and lex-parse. 
We used execution time, memory, and number of phrases as evaluation measures for each method.
The C++ programming language was used for implementing all the parsing algorithms. 
The implementations used in this experiment are available at \url{https://github.com/TNishimoto/lzrr}.
LZ77 and lex-parse were implemented in the standard manner and work 
in time and space linear to string length using $\SA, \ISA$, and $\LCP$ arrays.
For each method, we computed two sets of phrases for original string $T$ and reverse string $T^R$, respectively, and we took the set with the smaller number of phrases. 
We denote numbers of phrases as $|LZ77|$, $|LEX|$, and $|LZRR|$ for parsing algorithms 
of LZ77, lex-parse~(LEX), and LZRR, respectively. 
We performed all the experiments on one core of a 
quad-core Intel(R) Xeon(R) E5-2680 v2 (2.80 GHz) CPU with 256 GB of memory. 
\subsection{Results}
\begin{table}[t]
    \caption{The number of phrases for each method. 
    The smallest number of phrases for each string is depicted in bold.} 
    
    \label{table:exp1} 
    \scriptsize
    \center{	
    
    \begin{tabular}{|l|r||r|r|r||r|}
        \hline
       String & String length & $|LZ77|$ & $|LEX|$ & $|LZRR|$ & $\frac{|LZRR|}{|LZ77|}$ \\ \hline \hline
        fib41 & 267,914,296 & 22  & {\bf 4}  & 5  & 0.227  \\ \hline 
        rs.13 & 216,747,218 & 52  & {\bf 40}  & 51  & 0.981  \\ \hline
        tm29 & 268,435,456 & 56  & 43  & {\bf 31}  & 0.554  \\ \hline \hline
        dblp.xml.00001.1 & 104,857,600 & 59,385  & 58,537  & {\bf 55,127}  & 0.928  \\ \hline
        dblp.xml.00001.2 & 104,857,600 & 59,556  & 60,220  & {\bf 55,122}  & 0.926  \\ \hline
        dblp.xml.0001.1 & 104,857,600 & 78,167  & 82,879  & {\bf 73,584}  & 0.941  \\ \hline
        dblp.xml.0001.2 & 104,857,600 & 78,158  & 99,467  & {\bf 73,583}  & 0.941  \\ \hline
        sources.001.2 & 104,857,600 & 294,994  & 466,074  & {\bf 287,411}  & 0.974  \\ \hline
        dna.001.1 & 104,857,600 & 308,355  & 307,329  & {\bf 295,354}  & 0.958  \\ \hline
        proteins.001.1 & 104,857,600 & 355,268  & 364,024  & {\bf 337,711}  & 0.951  \\ \hline
        english.001.2 & 104,857,600 & 335,815  & 487,586  & {\bf 324,282}  & 0.966  \\ \hline \hline
        einstein.de.txt & 92,758,441 & 34,287  & 37,719  & {\bf 31,798}  & 0.927  \\ \hline
        einstein.en.txt & 467,626,544 & 89,437  & 96,487  & {\bf 83,368}  & 0.932  \\ \hline
        world\_leaders & 46,968,181 & 175,670  & 179,503  & {\bf 165,626}  & 0.943  \\ \hline
        influenza & 154,808,555 & 769,286  & 764,634  & {\bf 714,320}  & 0.929  \\ \hline
        kernel & 257,961,616 & 793,915  & 794,058  & {\bf 741,556}  & 0.934  \\ \hline
        cere & 461,286,644 & 1,695,631  & 1,649,448  & {\bf 1,597,657}  & 0.942  \\ \hline
        coreutils & 205,281,778 & 1,441,384  & 1,439,918  & {\bf 1,359,606}  & 0.943  \\ \hline
        Escherichia\_Coli & 112,689,515 & 2,078,512  & 2,014,012  & {\bf 1,961,296}  & 0.944  \\ \hline
        para & 429,265,758 & 2,332,657  & 2,238,362  & {\bf 2,200,802}  & 0.943  \\ \hline
    \end{tabular} 
    }
\end{table}

\begin{table}[htbp]
    \scriptsize
    \caption{The execution time and memory for each method.}
    \label{table:exp2} 
    \center{	

    \begin{tabular}{|l|r|r|r|r|r|r|r|}
        \hline
 &  & \multicolumn{3}{|c|}{Execution time [sec]} & \multicolumn{3}{|c|}{Memory consumption [MB]} \\ \hline
String & String length & LZ77 & LEX & LZRR & LZ77 & LEX & LZRR \\ \hline  \hline
einstein.de.txt & 92,758,441 & 24  & 16  & 27  & 2,266  & 2,266  & 3,808  \\ \hline
einstein.en.txt & 467,626,544 & 130  & 85  & 147  & 11,418  & 11,418  & 19,196  \\ \hline
world\_leaders & 46,968,181 & 8  & 5  & 16  & 1,148  & 1,148  & 1,939  \\ \hline
influenza & 154,808,555 & 42  & 27  & 51  & 3,781  & 3,781  & 6,351  \\ \hline
kernel & 257,961,616 & 71  & 47  & 88  & 6,299  & 6,299  & 10,602  \\ \hline
cere & 461,286,644 & 131  & 90  & 500  & 11,263  & 11,263  & 18,925  \\ \hline
coreutils & 205,281,778 & 56  & 37  & 68  & 5,013  & 5,013  & 8,453  \\ \hline
Escherichia\_Coli & 112,689,515 & 32  & 22  & 46  & 2,752  & 2,752  & 4,632  \\ \hline
para & 429,265,758 & 125  & 85  & 203  & 10,481  & 10,481  & 17,609  \\ \hline
    \end{tabular} 
    }
\end{table}
Table~\ref{table:exp1} shows the number of phrases for each method. 
The number of LZRR phrases was smaller than that of LZ77 phrases for all benchmark strings. 
Specifically, 
the number of LZRR phrases was approximately five percent smaller than that of LZ77 for all the strings except for fib41, rs.13, and tm29.
The number of LZRR phrases was smaller that of lex-parse phrases for most of the strings. 

Table~\ref{table:exp2} shows execution time and memory on limited benchmark strings for each method. 
The table for all the strings is presented in Appendix. 
Although our LZRR parsing needs $O(n^2 \alpha_{n^2}(n^2))$ time, 
the execution time was at most four times slower than that of LZ77 parsing.
This is because the number of while-loops in Algorithm~\ref{algo:lp} is 
much smaller than $n$ in practice. 
The memory for LZRR parsing was at most two times larger than that for LZ77 parsing. 
This is because the proposed algorithm needs the data structure for $\LF$ along with $\SA, \ISA$, and $\LCP$ arrays. 
\section{Conclusions}
We presented a new bidirectional parsing algorithm named Lempel-Zip 77 parsing with right reference (LZRR). 
The number of LZRR phrases is theoretically guaranteed to be smaller than that of LZ77. 
Experimental results using benchmark strings showed LZRR parsing works in practice. 
An interesting line of future work is to devise the LZRR parsing algorithm working in 
$o(n^2 \alpha_{n^2}(n^2))$ time or a compressed space.

\smallskip
\noindent \textbf{Acknowledgments.} We would like to thank Simon J. Puglisi for notifying us some related work~\cite{DBLP:conf/wea/DinklageFKLS17, DBLP:conf/latin/GagieNP18}.

\bibliographystyle{splncs03}
\bibliography{ref}

\clearpage
\section*{Appendix A: The proof of Lemma~\ref{lem:lf}}
\begin{algorithm}[htbp]
  \SetKwProg{Fn}{Function}{}{}
  Input: $P$, $L$, $j$ and $W$\; 
  Output: $\LF(P, j)$\;
  Create the union-find data structure for $\mathcal{D'} = \emptyset $\;
  $\ell \leftarrow 1$\;
  \While{$T[i+\ell-1] = T[j+\ell-1]$ and $\mathit{source}(i + \ell -1) =  \mathit{source}(j + \ell -1)$}{
      $\mathit{mUnion}(i + \ell -1, j + \ell -1)$\;
      $\ell \leftarrow \ell + 1$ \;
  }
  $\ell \leftarrow \ell - 1$ \;
  Delete the union-find data structure for $\mathcal{D'}$\;
  Initialize $W$\;
 \Return $\ell$\;
 
 \Fn{source($x$)}{
 $y \leftarrow \find(x)$ on $\mathcal{D}_{P}$\;
 \If{$i \leq y \leq i + \ell - 1$}{
 $z \leftarrow \find(W[y])$ on $\mathcal{D'}$\;
 \Return $z$\;
 }
 \Else{
 \Return $y$\;
 }
 }
 \Fn{mUnion($x, y$)}{
 $x' \leftarrow \mathit{source}(x)$\;
 $\makeset'(y)$\;
 \If{$W[x'] \not = -1 $}{
 $\union(W[x'], W[y])$ on $\mathcal{D'}$\;
 }
 \Else{
 $\makeset'(x')$\;
 $\union(x', y)$ on $\mathcal{D'}$\;
 }
 }
 \Fn{$\makeset'(x)$}{
   $p \leftarrow \makeset$ on $\mathcal{D'}$\;
   $W[x] \leftarrow p$\;
 }
 
  \caption{Modified $\LF$ algorithm.\label{algo:lf2}}
\end{algorithm}

\begin{figure}[htbp]
  \centerline{
		\includegraphics[width=0.7\textwidth]{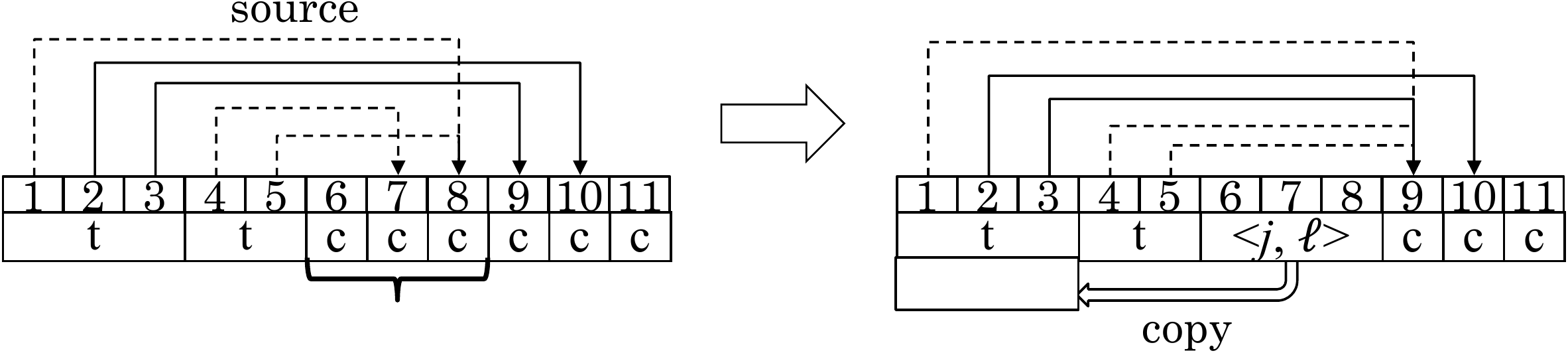} 
  }
	  \caption{ 
	  This figure represents the change of sources when a new target phrase is appended into a PBP. The ``t'' and  ``c'' in rectangles mean target phrases and character phrases, respectively.
	  }
 \label{fig:app1}
\end{figure}

To compute $\LF(P, j)$ without changing the union-find data structure $L$ for $\mathcal{D}_{P}$, 
we create an additional union-find data structure $L'$ and  
we emulate find operations on $\mathcal{D}_{P \cdot \langle j, \ell \rangle}$ using union-find data structures $L$ and $L'$ for $\mathcal{D}_{P}$. 
Since union operations are performed on $L'$, 
$L$ is not changed in $\LF$ algorithm. 

A key idea is that sources on the last phrase $\langle j, \ell \rangle$ are only changed by extending $P \cdot \langle j, \ell \rangle$. 
See Figure~\ref{fig:app1}. The left figure represents sources of positions on target phrases. 
The right figure represents the change of sources by appending new target phrase $\langle j, \ell \rangle$ to the target phrases. 
The new target phrase changes only sources on the phrase and these sources are determined by the phrase. 
This suggests that sources not on the target phrase $\langle j, \ell \rangle$ 
can be computed using $L$, and the other sources can be computed using $L$ and 
the additional union-find data structure $L'$ that manages sources of positions 
on the phrase $\langle j, \ell \rangle$. 
In addition, disjoint sets managed by $L'$ can be updated by union operations as Corollary~\ref{cor:da2}. 

Formally, let $U(P \cdot \langle j, \ell \rangle)$ be the set of positions on the phrase $\langle j, \ell \rangle$ and sources of those positions~(i.e., $U(P \cdot \langle j, \ell \rangle) = \{ i, \ldots , i+\ell-1 \} \cup \{ \vroot(P \cdot \langle j, \ell \rangle, x) \mid x \in \{ i, \ldots , i+\ell-1 \} \}$) and 
let $\mathcal{D'}_{P \cdot \langle j, \ell \rangle}$ be 
disjoint sets on $U(P \cdot \langle j, \ell \rangle)$ 
such that each set consists of all positions of the same source, 
where $i$ is the position following $P$. 
Let $\makeset'(x)$ be the operation on disjoint-sets $D$ that 
adds $\{ x \}$ into $D$ if $D$ does not contain $x$. 
Then the following lemma and corollary hold.
\begin{lemma}\label{lem:da_part4}
For a position $x \in \{ 1 ,2,  \ldots, n \}$,
if $\vroot(P, x) \not \in \{ i, i+1, \ldots, i+\ell -1 \}$ holds, 
then $\vroot(P, x) = \vroot(P \cdot \langle j, \ell \rangle, x)$ holds. 
Otherwise,  $\vroot(P \cdot \langle j, \ell \rangle, x') = \vroot(P \cdot \langle j, \ell \rangle, x)$ holds, where $x' = \vroot(P, x)$.
\end{lemma}
\begin{proof}
$T[r]$ is a character phrase on $B_{P}$ for each position $r \in \{ i, i+1, \ldots, i+\ell -1 \}$. 
If the source $x'$ of $x$ on $B_{P}$ is not a position in $\{ i, i+1, \ldots, i+\ell -1 \}$, 
then $x$ does not reach any position in $\{ i, \ldots, i+\ell -1 \}$. 
When a phrase is appended into $P$, 
the source $x'$ is changed if and only if the character phrase on $x'$ is changed. 
Thus, the source of $x$ is not changed by appending $\langle j, \ell \rangle$ into $P$, i.e., 
$\vroot(P, x) = \vroot(P \cdot \langle j, \ell \rangle, x)$. 

Otherwise, the source $x'$ is in $\{ i, \ldots, i+\ell -1 \}$ on $B_{P}$ and 
$x'$ has a source $x''$  on $B_{P \cdot \langle j, \ell \rangle}$ because $T[x']$ is not a character phrase on $B_{P \cdot \langle j, \ell \rangle}$. 
Since the source of $x'$ is that of $x$ on $B_{P \cdot \langle j, \ell \rangle}$, 
Lemma~\ref{lem:da_part4} holds. 

\end{proof}
\begin{corollary}\label{lem:da_part5}
(1) For an integer $x \in \{ i, \ldots, i + \ell -1 \}$, 
there exists a set $X \in \mathcal{D'}_{P \cdot \langle j, \ell \rangle}$ that 
contains two positions $x$ and $\vroot(P \cdot \langle j, \ell \rangle, x)$.
(2) $\mathcal{D'}_{P \cdot \langle j, \ell+1 \rangle}$ can be created by 
performing $O(1)$ $\union$ and $\makeset'$ operations on $\mathcal{D'}_{P \cdot \langle j, \ell \rangle}$. 
\end{corollary}
We compute the source of a given position on $\{ 1, \ldots, n \}$ by $O(1)$ find queries 
on $\mathcal{D}_{P}$ and $\mathcal{D'}_{P \cdot \langle j, \ell \rangle}$ using 
Lemma~\ref{lem:da_part4} and Corollary~\ref{lem:da_part5}. 
Note that we need to compute the position on the character phrase in 
a given set to obtain the source of a given position. 
For this reason, we use the position on a character phrase as the id of the set that contains the position. 
We can maintain such id using an additional array of length $m$ with the same time complexity, 
where $m$ is the cardinality of disjoint-sets. 

We also note that we need to convert integers in $\mathcal{D'}_{P \cdot \langle j, \ell \rangle}$ 
to consecutive integers. 
This is because disjoint sets of $L'$ are on consecutive integers 
since $\makeset$ creates the element $\{ m+1 \}$. 
Thus, we use an array $W$ of size $n$, 
where $W[x]$ stores the integer in $L'$ that corresponds to $x$ if $x \in U(P \cdot \langle j, \ell \rangle)$; otherwise $W[x]=-1$.
This array also enables us to emulate $\makeset'$ operations. 
Since the size of $W$ is $n$, 
we reuse $W$ during the LZRR parsing algorithm, and 
the algorithm creates the array in advance.
It takes $O(n)$ time and space.
$W$ can be initialized in $O(m)$ time, where $m$ is the number of positive integers in $W$.  

Algorithm~\ref{algo:lf2} shows the modified algorithm for computing $\LF(P, j)$ function 
using Lemma~\ref{lem:da_part4} and Corollary~\ref{lem:da_part5}. 
Algorithm~\ref{algo:lf2} computes $\LF(P, j)$ by $O(\ellmax)$ union and find operations and does not perform union operations on $\mathcal{D}_{P}$, where $\ellmax$ is the length of the longest valid phrase following $P$ with reference position $j$. 
As a result, Lemma~\ref{lem:lf} holds. 

Note that Algorithms~\ref{algo:lf} and~\ref{algo:lf2} can fail if
there exists an invalid PBP $P \cdot \langle j, \ell' \rangle$ for an integer $\ell' \in \{ 1, 2, \ldots, \ellmax \}$. 
If such an integer exists, then algorithms return $\ell' - 1$ and fail. 
However, such cases do not occur because 
we cannot remove infinite loops of references from an invalid PBP by appending phrases into the PBP. 

\section*{Appendix B: Computing $\SA_{k}[1..\ell]$ and $\LCE(k, \SA_{k}[1]), \ldots, \LCE(k, \SA_{k}[\ell])$ }
We show that we can compute $\SA_{k}[1..\ell]$ and $\LCE(k, \SA_{k}[1]), \ldots, \LCE(k, \SA_{k}[\ell])$ for a given $k$ and $\ell$ 
in $O(\ell)$ time using $T$ and $\SA, \ISA, \LCP$ arrays. 

We use the known fact that $\LCE(\SA[i], \SA[j]) = \min \{ \LCP[i+1], \ldots, \LCP[j] \}$ holds for two integers $1 \leq i < j \leq n$. 
When $\SA_{k}[1..\ell']$ stores the permutation of $\SA[i'..j']$ containing $k$ for some integer $\ell'$, 
$\SA_{k}[1..\ell'+1]$ can store $\SA[i'-1]$ or $\SA[j'+1]$ by the above fact, 
where $i'$ and $j'$ are integers such that $j' - i' + 1 = \ell$. 
Then $\SA_{k}[1..\ell+1]$ is also the permutation of a subarray of $\SA$ containing $k$. 
Thus, we compute $\SA_{k}[1..\ell]$ by using the above observation. 

We compute $\SA_{k}[\ell'+1]$ using $i', j', p$ and $q$, where 
$p = \LCE(k, \SA[i'])$ and $q = \LCE(k, \SA[j'])$. 
Since $\LCE(k, \SA[i'-1]) = \min \{ \LCP[i'], p \}$ and $\LCE(k, \SA[j'+1]) = \min \{ \LCP[j'+1], q \}$, 
$\SA_{k}[\ell'+1]$ can be computed in constant time. 
In addition, we can appropriately update the four parameters in constant time for $\SA_{k}[\ell'+2]$.
Therefore, we can compute $\SA_{k}[1..\ell]$ and $\LCE(k, \SA_{k}[1]), \ldots, \LCE(k, \SA_{k}[\ell])$ in $O(\ell)$ time and constant working space using a simple algorithm. 
\qed

\section*{Appendix C: The proof of the upper bound of LZRR phrases }
\begin{figure}[htbp]
  \centerline{
		\includegraphics[width=0.6\textwidth]{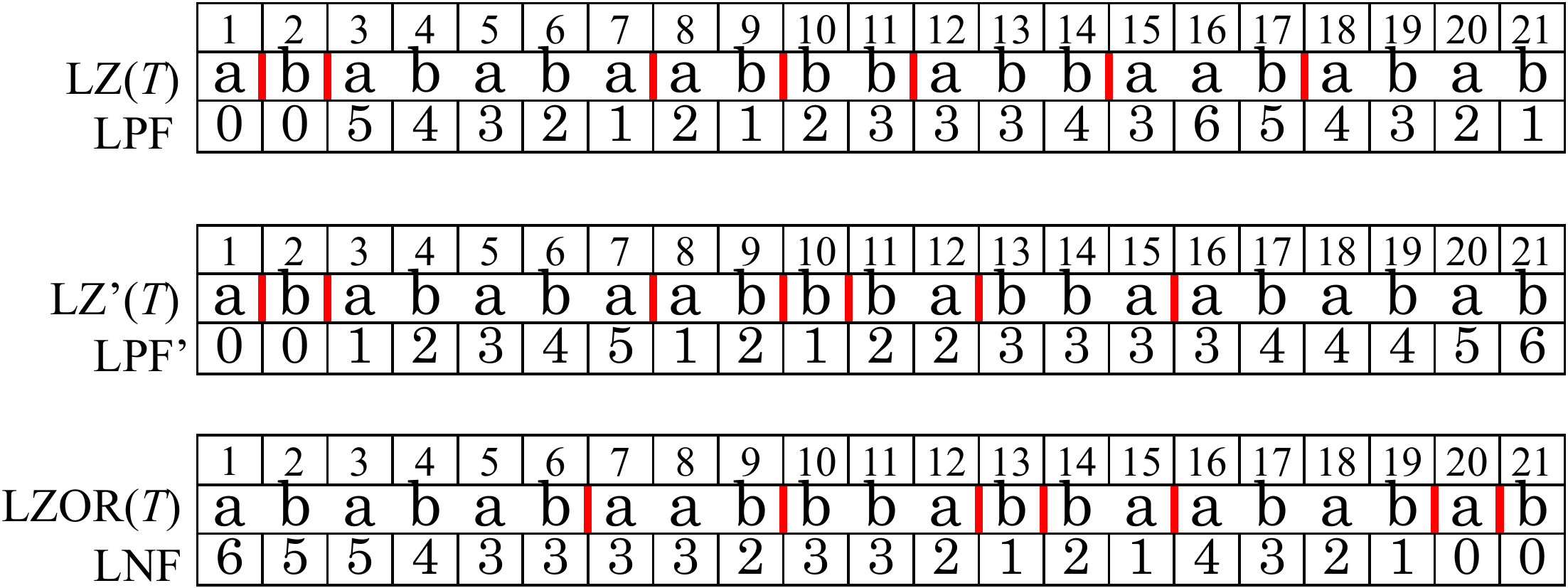} 
  }
	  \caption{ 
		  Examples of $\LZ(T)$, $\LZr(T)$ and $\LZOR(T)$. 
	  }
 \label{fig:lzor}
\end{figure}
We show three formulas using injective functions; 
for two BPs $F = f_{1}, f_{2}, \ldots, f_{k}$ and $F' = f'_{1}, f'_{2}, \ldots, f'_{k'}$ of $T$, 
if there exists an injective function $w$ that 
maps phrases in $F$ into distinct phrases in $F'$, 
then $k \leq k'$ holds.
In the remaining section, let $s_{x}$ and $e_{x}$~(resp. $s'_{x}$ and $e'_{x}$) be 
starting and ending positions of $x$-th phrase in $F$~(resp. $F'$).

\smallskip
\noindent \textbf{The proof of $|\LZr(T)| = |\LZ(T)|$.}
$\LZr(T) = f_{1}, f_{2}, \ldots, f_{k}$ parses greedily $T$ in the right-to-left order such that each phrase is the longest substring occurring previously~(left) in $T$. 
Formally, let $\LPF'$ be the integer array of length $n$ such that 
$\LPF'[i]$ stores the length of the longest substring of $T$ ending at position $i$ and occurring on $T[1..i-1]$ for all $i \in \{ 1, 2, \ldots , n \}$, 
i.e., $\LPF'[i] = \max (\{ 0 \} \cup \{\ell \mid \ell \in [1..i], |\Occ(T[..i-1], T[i-\ell+1..i])| > 0 \} )$. 
Then $\LZr(T) = f_{1}, \ldots, f_{k}$ is the valid BP of $T$ such that  
for all $x \in \{ 1, 2, \ldots , k \}$, 
the starting position $s_{x}$ of $f_{x}$ is $s_{x+1} - \max \{ 1,  \LPF'[s_{x+1}-1] \}$, where 
$s_{k+1} = n$. 
Figure~\ref{fig:lzor} illustrates examples of $\LZ(T)$ and $\LZr(T)$.

For $\LZ(T) = f_{1}, f_{2}, \ldots, f_{k}$ and $\LZr(T) = f'_{1}, f'_{2}, \ldots, f'_{k'}$, 
we define the function $w(x')$ that returns the integer $x$ such that 
$f_{x}$ contains $e'_{x'}$~(i.e., $s_{x} \leq e'_{x'} \leq e_{x}$ holds).
$w$ is injective if 
the starting position of each $\LZr$ phrase is not larger than that of the LZ phrase containing the ending position of the $\LZr$ phrase, i.e., $s'_{x'} \leq s_{w(x')}$ holds for all $x' \in \{ 1, 2, \ldots,  k' - 1 \}$. 
This is because no LZ phrases contain two ending positions in $\LZr$ phrases, i.e., 
no integer exists $y'$ such that $w(y') = w(y'+1)$ holds if $s'_{x'} \leq s_{w(x')}$ holds for all $x'$.

We show $s'_{x'} \leq s_{w(x')}$ using the substring $S$ starting at the starting position of $f_{w(x')}$ and 
ending at the ending position of $f'_{x'}$, i.e., $S = T[s_{w(x')}..e'_{x'}]$. 
When $|f'_{x'}| \geq |S|$ holds, $s'_{x'} \leq s_{w(x')}$ holds 
because $S$ is a suffix of $f'_{x'}$ and $S$ is a prefix of $f_{w(x')}$.
Thus, we show $|f'_{x'}| \geq |S|$ always holds.
If $S$ occurs in previously on $T$, then $|f'_{x'}| \geq |S|$
because $\LZr$ chooses the longest substring ending at position $e_{x'}$ and occurring on $T[1..e_{x'}-1]$.
Otherwise, $|S| = 1$ and $S = f'_{x'} = f_{w(x')}$ hold since $S$ is a new character, i.e., $\Occ(T[1..e'_{x'}-1], T[e'_{x'}]) = \emptyset$.  
Therefore, $s'_{x'} \leq s_{w(x')}$ holds for all $x' \in \{ 1, 2, \ldots, k' - 1 \}$. 

Similarly, $|\LZr(T)| \geq |\LZ(T)|$ holds by constructing the injective function 
that returns the integer $x'$ 
such that $f'_{x'}$ contains $s_{x}$ for a given $x$.

\smallskip
\noindent \textbf{The proof of $|\LZRR(T)| \leq |\LZOR(T)|$.}
$\LZOR(T) = f_{1}, f_{2}, \ldots, f_{k}$ parses $T$ in the left-to-right order such that each phrase is the longest substring occurring subsequently in $T$. 
Formally, let $\LNF$ be the integer array of length $n$ such that 
$\LNF[i]$ stores the length of the longest substring of $T$ starting at position $i$ and occurring on $T[i+1..]$ for all $i \in \{ 1, 2, \ldots , n \}$, i.e., $\LNF[i] = \max (\{ 0 \} \cup \{\ell \mid \ell \in [1..n-i+1], |\Occ(T[i+1..], T[i..i+\ell-1])| > 0 \}$). 
Then $\LZOR(T) = f_{1}, f_{2}, \ldots, f_{k}$ is the BP of $T$ such that 
for all $x \in \{ 2, 3, \ldots, k \}$, 
the starting position $s_{x}$ of $f_{x}$ is $s_{x-1} + \max \{ 1,  \LNF[s_{x-1}] \}$ and 
$s_{1} = 1$.
Figure~\ref{fig:lzor} illustrates an example of $\LZOR(T)$.

For $\LZRR(T) = f_{1}, f_{2}, \ldots, f_{k}$ and $\LZOR(T) = f'_{1}, f'_{2}, \ldots, f'_{k'}$, 
let $w(x)$ be the function that returns the integer $x'$ such that 
$f'_{x'}$ contains $s_{x}$.
Then $w$ is injective if $w(x) < w(x+1)$ holds for all $x \in \{ 1, 2,  \ldots ,k-1 \}$. 
We use the following lemma. 
\begin{lemma}\label{lem:right_ref}
    Let $P = f_{1}, \ldots , f_{b}$ be a valid PBP of $T$. 
    Then $P \cdot \langle j, \ell \rangle $ is also valid  
    for any right target phrase $\langle j, \ell \rangle$, i.e., $T[i..i+\ell-1] = T[j..j+\ell-1]$ and $j > i$ hold, where $i = |f_{1} \cdots f_{b}| + 1$. 
\end{lemma}
\begin{proof}
 $T[i..n]$ are represented character phrases on $B_{P}$ since $T[i..n]$ has not been parsed. 
 This means that $ \vroot(P, i+\ell'-1) \not = \vroot(P, j+\ell'-1)$ for any $\ell' \in \{ 1, 2, \ldots, \ell \}$. 
 Therefore $P \cdot \langle j, \ell' \rangle$ is valid by Corollary~\ref{cor:da2}. 
\end{proof}
$w(x) < w(x+1)$ holds if $\max \{1, \LNF[s_{x}] \} \leq |f_{x}|$ holds for all $x \in \{ 1, 2, \ldots k \}$. 
Recall that $\LP$ function returns the valid longest bidirectional phrase.  
The $\LNF$ array and Lemma~\ref{lem:right_ref} 
suggest that the length of the phrase of $\LZRR(T)$ starting at position $i$ 
is at least $\max \{1, \LNF[i] \}$. 
Thus $w(x) < w(x+1)$ holds for all $x$, 
$w$ is injective, and hence $|\LZRR(T)| \leq |\LZOR(T)|$ holds.

\smallskip
\noindent \textbf{The proof of $|\LZOR(T)| = |\LZr(T^R)|$.}
$|\LZOR(T)| = |\LZ'(T^{R})|$ holds clearly because 
$\LPF'[x] = \LNF_{T^R}[n-x+1]$ holds for all $x \in \{ 1,2,  \ldots, n \}$, 
where $\LNF_{T^R}$ is the $\LNF$ array of $T^{R}$.

\section*{Appendix D: The proof of Lemma~\ref{lem:maxlf} }
\begin{proof}
Recall that $\LCE(i, \SA_{i}[1]) \geq \cdots \geq \LCE(i, \SA_{i}[n])$ holds. 
On the other hand, $\LF(P, j) \leq \LCE(i, j)$ holds for all $j \in \{ 1 , 2, \ldots, n \}$ because 
$\LF(P, j)$ represents the length of the common prefix of $T[i..]$ and $T[j..]$. 
Therefore, $\ellmax = \ell_{n} = \ell_{\kmax}$ holds, which means 
at least one position $j'$ exists such that $\LF(P, j') = \ellmax$ in $\SA_{i}[..\kmax]$. 
\end{proof}

\section*{Appendix E: Experiments }

\begin{table}[htbp]
    \scriptsize
    \caption{The full version of Table~\ref{table:exp2}.}
    \label{table:exp3} 
    \center{	

    \begin{tabular}{|l|r|r|r|r|r|r|r|}
        \hline
 &  & \multicolumn{3}{|c|}{Execution time [sec]} & \multicolumn{3}{|c|}{Memory consumption [MB]} \\ \hline
String & String length & LZ77 & LEX & LZRR & LZ77 & LEX & LZRR \\ \hline  \hline
fib41 & 267,914,296 & 99  & 74  & 113  & 6,542  & 6,542  & 11,978  \\ \hline
rs.13 & 216,747,218 & 79  & 59  & 110  & 5,292  & 5,293  & 9,654  \\ \hline
tm29 & 268,435,456 & 108  & 81  & 142  & 6,554  & 6,555  & 11,797  \\ \hline
dblp.xml.00001.1 & 104,857,600 & 30  & 21  & 42  & 2,561  & 2,561  & 4,308  \\ \hline
dblp.xml.00001.2 & 104,857,600 & 30  & 20  & 41  & 2,561  & 2,561  & 4,305  \\ \hline
dblp.xml.0001.1 & 104,857,600 & 30  & 20  & 42  & 2,561  & 2,561  & 4,303  \\ \hline
dblp.xml.0001.2 & 104,857,600 & 30  & 20  & 41  & 2,561  & 2,561  & 4,303  \\ \hline
sources.001.2 & 104,857,600 & 28  & 19  & 41  & 2,561  & 2,561  & 4,302  \\ \hline
dna.001.1 & 104,857,600 & 30  & 20  & 41  & 2,561  & 2,561  & 4,302  \\ \hline
proteins.001.1 & 104,857,600 & 31  & 21  & 42  & 2,561  & 2,561  & 4,302  \\ \hline
english.001.2 & 104,857,600 & 30  & 21  & 42  & 2,561  & 2,561  & 4,302  \\ \hline
einstein.de.txt & 92,758,441 & 24  & 16  & 27  & 2,266  & 2,266  & 3,808  \\ \hline
einstein.en.txt & 467,626,544 & 130  & 85  & 147  & 11,418  & 11,418  & 19,196  \\ \hline
world\_leaders & 46,968,181 & 8  & 5  & 16  & 1,148  & 1,148  & 1,939  \\ \hline
influenza & 154,808,555 & 42  & 27  & 51  & 3,781  & 3,781  & 6,351  \\ \hline
kernel & 257,961,616 & 71  & 47  & 88  & 6,299  & 6,299  & 10,602  \\ \hline
cere & 461,286,644 & 131  & 90  & 500  & 11,263  & 11,263  & 18,925  \\ \hline
coreutils & 205,281,778 & 56  & 37  & 68  & 5,013  & 5,013  & 8,453  \\ \hline
Escherichia\_Coli & 112,689,515 & 32  & 22  & 46  & 2,752  & 2,752  & 4,632  \\ \hline
para & 429,265,758 & 125  & 85  & 203  & 10,481  & 10,481  & 17,609  \\ \hline
    \end{tabular} 
    }
\end{table}

\end{document}